\documentclass[sigconf]{acmart}

\usepackage{bbm} 
\usepackage{algorithm}
\usepackage{url}
\usepackage{algpseudocode}
\def\BSTATE{\STATE\hskip-\ALG@thistlm}
\makeatother

\usepackage[caption=false,font=footnotesize]{subfig}

\def\nb0{{\mathbf{0}}}
\def\nb1{{\mathbf{1}}}

\newtheorem{lemma}{Lemma}

\newtheorem{theorem}{Theorem}

\newtheorem{remark}{Remark}

\def\E{\mathbb{E}}

\def\P{\mathbb{P}}

\def\g{\left.\right|}

\AtBeginDocument{%
  }

\copyrightyear{2025}
\acmYear{2025}
\setcopyright{rightsretained}
\acmConference[MobiHoc '25]{The Twenty-sixth International Symposium on
Theory, Algorithmic Foundations, and Protocol Design for Mobile Networks and
Mobile Computing}{October 27--30, 2025}{Houston, TX, USA}
\acmBooktitle{The Twenty-sixth International Symposium on Theory,
Algorithmic Foundations, and Protocol Design for Mobile Networks and Mobile
Computing (MobiHoc '25), October 27--30, 2025, Houston, TX, USA}
\acmDOI{10.1145/3704413.3764460}
\acmISBN{979-8-4007-1353-8/2025/10}

\begin{document}

\title{Online Learning for Optimizing AoI-Energy Tradeoff under Unknown Channel Statistics}

\author{Mohamed A. Abd-Elmagid}
\affiliation{%
  \institution{Department of ECE \\
  The Ohio State University}
  \city{Columbus}
  \state{OH}
  \country{USA}
}
\email{abd-elmagid.1@osu.edu}

\author{Ming Shi}
\affiliation{%
  \institution{Department of EE\\ University at Buffalo}
  \city{Buffalo}
  \state{NY}
  \country{USA}}
\email{mshi24@buffalo.edu}

\author{Eylem Ekici}
\affiliation{%
  \institution{Department of ECE \\
  The Ohio State University}
  \city{Columbus}
  \state{OH}
  \country{USA}
}
\email{ekici.2@osu.edu}

\author{Ness B. Shroff}
\affiliation{%
  \institution{Departments of ECE and CSE \\
  The Ohio State University}
  \city{Columbus}
  \state{OH}
  \country{USA}
}
\email{shroff.11@osu.edu}


\renewcommand{\shortauthors}{Mohamed A. Abd-Elmagid, Ming Shi, Eylem Ekici, and Ness B. Shroff}

\begin{abstract}
 We consider a real-time monitoring system where a source node (with energy limitations) aims to keep the information status at a destination node as fresh as possible by scheduling status update transmissions over a set of channels. The freshness of information at the destination node is measured in terms of the Age of Information (AoI) metric. In this setting, a natural tradeoff exists between the transmission cost (or equivalently, energy consumption) of the source and the achievable AoI performance at the destination. This tradeoff has been optimized in the existing literature under the assumption of having a complete knowledge of the channel statistics. In this work, we develop online learning-based algorithms with finite-time guarantees that optimize this tradeoff in the practical scenario where the channel statistics are unknown to the scheduler. In particular, when the channel statistics are known, the optimal scheduling policy is first proven to have a threshold-based structure with respect to the value of AoI (i.e., it is optimal to drop updates when the AoI value is below some threshold). This key insight was then utilized to develop the proposed learning algorithms that surprisingly achieve an order-optimal regret (i.e., $O(1)$) with respect to the time horizon length.
\end{abstract}

%
%



\maketitle

\section{Introduction} \label{sec:intro}
Timely delivery of real-time status updates is necessary for many critical and emerging applications including healthcare, factory automation, intelligent transportation systems, and smart homes, to name a few. A typical real-time status update system
consists of an energy-constrained source node (e.g., a small sensor) that generates status updates about some physical process of interest, and then sends them through a communication system to a destination node. Clearly, excessive transmissions of status updates can maintain the freshness of information available at the destination at the price of quickly exhausting the limited energy available at the source. Therefore, there exists a natural tradeoff between maintaining the freshness of information available at the destination and the transmission cost (or energy cost) of the source. Scheduling the transmissions of status updates to optimize this tradeoff is challenging especially since in practice the statistics of the channels between the source and destination nodes are often unknown to the scheduler. In this paper, we address this open problem by developing novel online learning-based scheduling algorithms with provable guarantees. 

We employ AoI as a metric to quantify the freshness of information at the destination about the process observed by the source. Specifically, AoI is defined as the time elapsed since the last successfully received status update at the destination was generated at the source \cite{kaul2012real}. There have been two main research directions in the AoI research area. The first direction aimed to analyze/characterize AoI in different queueing-theoretic models/disciplines, and the second direction was focused on the optimization of AoI in different communication systems that deal with time-sensitive information. Interested readers are advised to refer to \cite{pappas2023age} for a comprehensive book and \cite{roy_survey} for a recent survey. Since this paper belongs to the second research direction, we next discuss the most closely-relevant prior optimization-based studies on AoI.

The scheduling problem to minimize AoI in single-hop wireless networks with unreliable channels was studied in \cite{kadota2018scheduling,kadota2018optimizing,hsu2018age,hsu2019scheduling,tripathi2019whittle,sun2019closed}. In particular, the problem was formulated as a restless multi-armed bandit (MAB) problem for which Whittle Index-based scheduling policies were developed. A common assumption considered in \cite{kadota2018scheduling,kadota2018optimizing,hsu2018age,hsu2019scheduling,tripathi2019whittle,sun2019closed} was that the statistics of the channels and/or the channel state information are known to the scheduler. Also, none of these studies accounted for the energy limitations at the source node(s). Further, these prior works have mostly been focused on the study of the infinite horizon model, whereas we develop in this paper scheduling algorithms with provable finite-time guarantees.

For the case when the channel statistics are unknown to the scheduler, online learning-based scheduling algorithms with provable finite-time guarantees were developed in \cite{fatale2021regret,juneja2021correlated,prasad2021decentralized,atay2021aging}. The authors of \cite{fatale2021regret} considered a system setting where the source is connected to the destination through a set of unreliable channels (i.e., each channel is associated with a different reliability or successful transmission probability). The study in \cite{fatale2021regret} was extended in \cite{juneja2021correlated} to the case of having correlated unreliable channels, in \cite{prasad2021decentralized} to the multi-source setting where each source generates a status update every time slot, and in \cite{atay2021aging} to the multi-source setting with random status update arrivals at different sources. The scheduling of status updates over different channels was formulated as a multi-armed bandit problem in \cite{fatale2021regret,juneja2021correlated,prasad2021decentralized,atay2021aging} where each channel corresponds to one arm. The reward obtained from selecting one arm in some time slot is a function of the reliability of that arm and the AoI value at the beginning of that time slot (without accounting for the transmission cost of sending status updates). The regret of UCB \cite{auer2002finite} and Q-UCB \cite{krishnasamy2016regret} algorithms (with respect to the optimal policy that knows the channel statistics a priori) were shown in \cite{fatale2021regret} to scale as $O({\rm log} T)$ and $O({\rm log}^3 T)$, respectively, where $T$ is the time horizon length. The authors of \cite{prasad2021decentralized} developed a UCB-based distributed learning algorithm that scales as $O(\rm{log}^2 T)$, and the authors of \cite{atay2021aging} utilized the knowledge about the system being empty (i.e., there are no status updates to transmit) or not to develop a learning algorithm that achieves a bounded regret with respect to $T$ (i.e., $O(1)$) when the arrival rates at different sources are relatively small.

A key distinction between \cite{fatale2021regret,juneja2021correlated,prasad2021decentralized,atay2021aging} and this paper is the structure of the optimal policy to which the proposed learning algorithms are compared (to obtain the regret). In particular, the optimal policy for the settings studied in \cite{fatale2021regret,juneja2021correlated,prasad2021decentralized,atay2021aging} is to simply send a status update over the channel with the highest successful transmission probability every time slot (whenever the system is not empty), and hence the scheduling problem could be formulated as a multi-armed bandit problem. Since this paper accounts for the transmission costs of sending status updates over different channels, the simple structure of the optimal policy in \cite{fatale2021regret,juneja2021correlated,prasad2021decentralized,atay2021aging} does not hold here anymore. In particular, it may be optimal in our setting to remain idle in some time slots (and drop the generated status updates). Thus, the decision of sending a status update should also depend on the AoI value, and hence the multi-armed bandit problem formulation in \cite{fatale2021regret,juneja2021correlated,prasad2021decentralized,atay2021aging} is not sufficient to study the scheduling problem considered in this paper. This key difference between the structures of the optimal policies has significant impact on the development of the learning algorithms in this paper and makes the regret analysis much more challenging. Before going into more details about our contributions, it is instructive to note that scheduling problems to jointly optimize AoI and transmission cost or other costs have been studied in a variety of settings \cite{tseng2019online,fountoulakis2020optimal,bedewy2021optimal,liu2023toward,saurav2021minimizing,abd2019tcom,tripathi2021online,liu2024learning}. However, none of these works considered that the channel statistics are unknown to the scheduler, and most of them were focused on the study of the infinite horizon model.

{\it Contributions.} This paper presents novel online learning-based scheduling algorithms with provable finite-time guarantees to optimize AoI for energy-constrained communications under unknown channel statistics. In particular, we study a system setting in which an energy-constrained source node is connected to a destination node through a set of channels, where the channel statistics are assumed to be unknown to the scheduler. Towards developing AoI-aware online learning-based algorithms for this setting, we first analyze the structure of the optimal policy (that knows the channel statistics a priori) for the infinite time average-cost problem. Specifically, this optimal policy is proven to have a threshold-based structure with respect to the value of AoI (i.e., it is optimal to drop updates when the AoI is below some threshold). This key insight is then utilized to develop the proposed learning algorithms for the finite-time horizon model under consideration. 
Our proposed AoI-aware learning algorithms (with and without an exploration bonus) are proven to surprisingly have a bounded regret performance with respect to the time horizon length (i.e., $O(1)$). Extensive simulations are conducted to show the impact of different system design parameters on the empirical performance of the proposed learning algorithms. \textit{To the best of our knowledge, this paper makes the first attempt towards developing AoI-aware learning algorithms with a provable order-optimal regret performance for optimizing the fundamental AoI-energy tradeoff.}

\section{System Model and Problem Statement}\label{sec:Model}
\subsection{Network Model}
We consider a real-time monitoring system where a source node is connected to a destination node through $C$ communication channels ($C_i$ denotes the $i$-th channel). Without loss of generality, we consider a discrete time finite horizon composed of $T$ slots of
unit length. Hence, the terms power and energy are used interchangeably throughout the paper. At the beginning of each time slot, the source generates a fresh status update, and either transmits it to the destination using one of the channels or drops it. A power cost $P$ is associated with each transmission attempt over any of the channels, and the transmission power cost of time slot $t$ is denoted by $P(t)$. Note that $P(t)$ is equal to zero when the status update generated at the beginning of time slot $t$ is dropped, and is equal to $P$ otherwise. The freshness of information at the destination node is measured using the AoI metric. In particular, the AoI measures the time elapsed since the generation time of the latest successfully received status update at the destination node. Let $A(t)$ denote the AoI value at the beginning of time slot $t$. Without loss of generality, we assume that $A(t)$ is upper bounded by a finite value $A_{\rm m}$ which can be chosen to be arbitrarily large. When
$A(t)$ reaches $A_{\rm m}$, it means that the available information at the destination node is too stale to be of
any use. A status update transmission over channel $C_i$ is successful with probability $\mu_i$, independent of all other channels and
across time slots. The values of $\{\mu_i\}$ are assumed to be unknown to the scheduler. The total cost of time slot $t$ is defined as 
\begin{align}\label{cost_slot}
    C(t) = \alpha A(t) + (1 - \alpha) P(t),
\end{align}
where $\alpha \in [0,1]$. Our intention behind using a weighted cost function \footnote{Note that the results obtained in this paper (for the structure of the optimal policy in Section \ref{sec:known} as well as the regret bounds in Section \ref{sec:learningopt}) using the linear age cost function in (\ref{cost_slot}) can be readily extended to the case of having a non-decreasing age function $\mathcal{F}(A(t))$, i.e., $C(t) = \alpha \mathcal{F}(A(t)) + (1 - \alpha) P(t)$. In Section \ref{sec:learningopt}, we also provide numerical results demonstrating the bounded regret performance of our proposed order-optimal learning algorithm for a non-linear age cost function.} is to provide a generic problem formulation that allows the scheduler to set the importance weights of AoI and power consumption in the optimization problem. 

\textit{State and action spaces}. At the beginning of time slot $t$, the state of the system $s(t)$ is represented by the AoI value $A(t)$, i.e., $s(t) = A(t) \in \mathcal{S} = \{1, 2, \cdots , A_{\rm m}\}$. Based on the state $s(t)$, the action taken in slot $t$ is given by $a(t) \in \mathcal{A} = \{0, 1, \cdots, C\}$. In particular, when $a(t) = 0$, the status update generated by the source at the beginning of slot $t$ is dropped, and $A(t + 1) = {\rm min}\left \{A_{\rm m}, A(t)+1 \right \}$. On the other hand, when $a(t) = i > 0$, the generated status update is transmitted over channel $C_i$. Further, $A(t + 1)$ is given by
\begin{align}\label{eq:AoI_evol}
A(t + 1) = \begin{cases}
\begin{aligned}
&1,\; &&\text{with probability}\; \mu_i,\\
&{\rm min}\left \{A_{\rm m}, A(t)+1 \right \},\; &&\text{with probability}\; 1 - \mu_i.
\end{aligned}
\end{cases}
\end{align}

Based on the above definitions, the total cost of time slot $t$ in (\ref{cost_slot}) can be expressed as
\begin{align}\label{cost_reexpressed}
    C(s(t),a(t)) = \alpha s(t) + (1 - \alpha) P \mathbbm{1}\left(a\left(t\right) \neq 0\right),
\end{align}
where $\mathbbm{1}(\cdot)$ is the indicator function.

\subsection{Problem Statement}
A policy $\pi = \{\pi_1, \pi_2, \cdots, \pi_{T}\}$ is a sequence of mappings from the state space to the action space over different time slots, i.e., $\pi_i : \mathcal{S} \rightarrow \mathcal{A}, \forall i$. We also use $\pi^{*} = \{\pi^*_1, \pi^*_2, \cdots, \pi^*_{T}\}$ to refer to the optimal policy, which has a complete knowledge of the statistics of the channels (or the probabilities $\{\mu_i\}$) a priori. In absence of any knowledge about $\{\mu_i\}$, the accumulated cost over $T$ time slots under a policy $\pi$ starting from state $s$ is given by 
\begin{align}\label{cost_accum}
C(\pi,s,T) = \sum_{t=1}^{T}{C(t)}.
\end{align}

 The total regret of a policy $\pi$ with respect to $\pi^*$ after $T$ time slots is defined as 
\begin{align}\label{reg}
R^{\pi}(T) = \mathbb{E}[C(\pi,s,T)] - \mathbb{E}[C(\pi^*,s,T)],
\end{align}
where the expectation is taken with respect to the statistics of the channels. Our goal is to develop a learning algorithm which determines $\pi$ such that a tight upper bound on the total regret in (\ref{reg}) is obtained.

\section{The Case When the Statistics of the Channels Are Known}\label{sec:known}
The first step towards developing a learning algorithm that achieves a tight upper bound on the total regret in (\ref{reg}) is to understand the structure of the optimal policy $\pi^*$ (that has a complete knowledge of the successful transmission probabilities over different channels $\{\mu_i\}$). Although $\pi^*$ can be evaluated using the standard backward induction algorithm, it is not possible to obtain analytical insights about its structure in the finite time horizon problem under consideration. Because of that, we first characterize the structure of $\pi^*$ for the infinite time average-cost problem in this section, and then utilize the obtained insights to develop the learning algorithms for the finite time horizon problem (when the probabilities $\{\mu_i\}$ are unknown) in the next section. Specifically, the expected long-term average cost under policy $\pi$ can be expressed as 
\begin{align}\label{cost_average}
\rho(\pi,s) \!=\! \lim_{T\to \infty} \frac{1}{T} \mathbb{E}[C(\pi,s,T)].
\end{align}

Due to the nature of the evolution of AoI (as described by (\ref{eq:AoI_evol})), and the independence of channel statistics over time, the problem can be modeled as an infinite horizon average-cost
MDP with finite state and action spaces $\mathcal{S}$ and $\mathcal{A}$, respectively. Since there exists an optimal stationary deterministic policy (minimizing $\rho(\pi,s)$) for solving MDPs with finite state and action spaces \cite{bertsekas2011dynamic}, we aim at investigating the structure of that stationary deterministic policy in the sequel.

\begin{lemma}\label{Lemma:1}
 The stationary deterministic optimal policy $\pi^\star$ can be evaluated by solving the following Bellman's equations for average-cost MDPs \cite{bertsekas2011dynamic}:
\begin{align}\label{belman_equation}
\rho^* + V(s) = \underset{a \in \mathcal{A}}{\rm min}\; Q(s,a), s \in \mathcal{S},
\end{align}
where $\rho^* = \underset{\pi}{\rm min}\; \rho(\pi,s)$, $V(s)$ is the value function, and $Q(s,a)$ is the $Q$-function $\big($also referred to as the $Q$-factors, $\forall s \in \mathcal{S}$ and $a \in \mathcal{A}\big)$, which is the expected cost resulting from taking action $a$ in state $s$, i.e.,
\begin{align}\label{Q_func}
Q(s,a) = \alpha s + (1 - \alpha) P \mathbbm{1}\left(a \neq 0\right) + \sum\limits_{s' \in \mathcal{S}} {\P(s' \g s, a) V(s')},
\end{align}
where $\P(s' \g s, a)$ is the transition probability of moving from state $s$ to state $s'$ as a result of taking action $a$, which can be evaluated from (\ref{eq:AoI_evol}) as
\begin{align}\label{trans_prob}
 \P(s' \g s, a) = \begin{cases}
\begin{aligned}
&1,\; && s' = s^{+}\;{\rm and}\;a=0,\\
& 1 - \mu_i,\; &&s' = s^{+}\;{\rm and}\;a=i > 0,\\
&\mu_i,\; &&s' = 1\;{\rm and}\;a= i > 0,\\
&0,\; &&{\rm otherwise},
\end{aligned}
\end{cases}
\end{align}
where $s^{+} = {\rm min}\left \{A_{\rm m}, s+1 \right \}$. In addition, the optimal action taken at state $s$ is given by
\begin{align}\label{eq:optimal_policy}
\pi^*(s) = {\rm arg} \underset{a \in \mathcal{A}}{\rm min}\; Q(s,a).
\end{align}
\end{lemma}

Since the weak accessibility condition holds for our problem, a solution for the Bellman's equations in Lemma \ref{Lemma:1} is
guaranteed to exist \cite{bertsekas2011dynamic}. We will now analytically characterize the structure of the stationary deterministic optimal policy $\pi^*$ using the Value Iteration Algorithm (VIA). According to the VIA, the value function $V(s)$ can be evaluated iteratively such that $V(s)$ at iteration $n$, $n = 1, 2, \cdots$, is computed as 
\begin{align}\label{value_function_itern}
 \nonumber& V(s)^{(n)}= \underset{a \in \mathcal{A}}{\rm min}\; Q(s,a)^{(n - 1)}, \\
 &= \underset{a \in \mathcal{A}}{\rm min} \alpha s + (1 - \alpha) P \mathbbm{1}\left(a \neq 0\right) + \sum\limits_{s' \in \mathcal{S}} {\P(s' \g s, a) V(s')^{(n - 1)}},
\end{align}
where $s \in \mathcal{S}$. Hence, the optimal policy at iteration $n$ is given by
\begin{align}\label{policy_itern}
\pi^{*(n)}(s) = {\rm arg} \underset{a \in \mathcal{S}}{\rm min}\; Q(s,a)^{(n - 1)}.
\end{align}

As per the VIA, under any initialization of the value function $V(s)^{(0)}$, the sequence $\left \{ V(s)^{(n)}\right\}$ converges to $V(s)$ which satisfies the Bellman's equation in (\ref{belman_equation}), i.e.,
\begin{align}\label{conv}
\underset{n \rightarrow \infty}{\rm lim} V(s)^{(n)} = V(s).
\end{align}

Based on the VIA, the following Lemma characterizes the monotonicity property of the value function with respect to the system state. 
\begin{lemma}\label{Lemma:2}
The value function $V(s)$, satisfying the Bellman's equation in (\ref{belman_equation}) and corresponding to the optimal policy $\pi^*$, is non-decreasing with respect to $s$.
\end{lemma}
\begin{proof}
Consider two states $s_1$ and $s_2$ such that $s_1 \leq s_2$. Hence, the objective is to show that $V(s_1) \leq V(s_2)$. According to (\ref{conv}), it is then sufficient to show that $V(s_1)^{(n)} \leq V(s_2)^{(n)}, \forall n$, which we prove using mathematical induction. Particularly, the relation holds by construction for $n = 0$ since it corresponds to the initial values for the value function which can be chosen arbitrary. Now, we assume that $V(s_1)^{(n)} \leq V(s_2)^{(n)}$ holds for some $n$, and then show that it holds for $V(s_1)^{(n + 1)} \leq V(s_2)^{(n + 1)}$ as well. Particularly, according to (\ref{value_function_itern}) and (\ref{policy_itern}), $V(s_1)^{(n + 1)}$ and  $V(s_2)^{(n + 1)}$ can be respectively expressed as 
\begin{align}\label{v_s1_n+1}
\nonumber V(s_1)^{(n + 1)} & = \alpha s_1 + (1 - \alpha) P \mathbbm{1}\left(\pi^{* (n)} (s_1) \neq 0\right) \\ \nonumber&+ \sum\limits_{s' \in \mathcal{S}} {\P(s' \g s_1, \pi^{* (n)} (s_1)) V(s')^{(n)}}, \\
 \nonumber& \overset{\rm (a)}{\leq} \alpha s_1 + (1 - \alpha) P \mathbbm{1}\left(\pi^{* (n)} (s_2) \neq 0\right) \\&+ \sum\limits_{s' \in \mathcal{S}} {\P(s' \g s_1, \pi^{* (n)} (s_2)) V(s')^{(n)}},
\end{align}
\begin{align}\label{v_s2_n+1}
V(s_2)^{(n + 1)} \nonumber&= \alpha s_2 + (1 - \alpha) P \mathbbm{1}\left(\pi^{* (n)} (s_2) \neq 0\right) \\&+ \sum\limits_{s' \in \mathcal{S}} {\P(s' \g s_2, \pi^{* (n)} (s_2)) V(s')^{(n)}},
\end{align} 
where step (a) follows since it is not optimal to take action $\pi^{*(n)}(s_2)$ in state $s_1$. From (\ref{trans_prob}), note that we have
\begin{align}\label{costtogo}
 \nonumber& \sum\limits_{s' \in \mathcal{S}} {\P(s' \g s_i, a) V(s')^{(n)}}  = \mathbbm{1}\left(a = 0\right) V(s_i^+)^{(n)} \\&+ \mathbbm{1}\left(a \neq 0\right)\left[\mu_a V(1)^{(n)} + \left(1 - \mu_a\right)V(s_i^+)^{(n)}\right].
\end{align} 

Since $s_1 \leq s_2$, we have $V(s_1^+) \leq V(s_2^+)$, and hence we observe from (\ref{costtogo}) that 
\begin{align*}
\sum\limits_{s' \in \mathcal{S}} {\left[\P(s' \g s_1, \pi^{* (n)} (s_2)) - \P(s' \g s_2, \pi^{* (n)} (s_2)) \right] V(s')^{(n)}} \leq 0. 
\end{align*}

Thus, the right hand side of (\ref{v_s1_n+1}) is less than or equal to $V(s_2)^{(n+1)}$, which leads to having $V(s_1)^{(n+1)} \leq V(s_2)^{(n+1)}$. This completes the proof.
\end{proof}

Let $k^*$ denote the index of the channel with the highest successful transmission probability, i.e., $\mu_{k^*} > \mu_k, \forall k \in \mathcal{A} \setminus \{k^*\}$. Based on Lemma \ref{Lemma:2}, the following Lemma characterizes the structure of the optimal policy $\pi^*$.

\begin{lemma}\label{Lemma:3}
The optimal policy $\pi^*$ has the following structural properties: \\(i) When $s_1 \geq s_2$, if $\pi^*(s_1) = 0$, then $\pi^*(s_2) = 0$. \\(ii) When $s_1 \leq s_2$, if $\pi^*(s_1) = k^*$, then $\pi^*(s_2) = k^*$.
\end{lemma}
\begin{proof}
 We start the proof by showing that $\pi^*(s) \in \{0,k^*\}, \forall s \in \mathcal{S}$. In particular, for $k > 0$, we have
 \begin{align}\label{Q_func_0}
Q(s,0) = \alpha s + V(s^+),
\end{align}
 \begin{align}\label{Q_func_a}
Q(s,k) = \alpha s + V(s^+) + (1 - \alpha) P - \mu_k \left[V(s^+) - V(1)\right].
\end{align}

From Lemma \ref{Lemma:2}, we have $V(s^+) - V(1) \geq 0$, and hence we observe from (\ref{Q_func_a}) that $k^* = {\rm arg} \underset{k \in \mathcal{A}\setminus \{0\}}{\rm min}\; Q(s,k), \forall s \in \mathcal{S}$. Hence, $\pi^*(s) \in \{0,k^*\}, \forall s \in \mathcal{S}$. Now, note that proving that $\pi^*(s_1) = a$ leads to $\pi^*(s_2) = a'$ is equivalent to showing that
\begin{align}\label{struc_prop}
Q(s_2,a) - Q(s_2,a') \leq Q(s_1,a) - Q(s_1,a'), \forall a' \neq a,
\end{align}
where this holds since if $a$ is optimal in state $s_1$, then we have $Q(s_1,a) - Q(s_1,a') \leq 0, \forall a'$, which leads to $Q(s_2,a) \leq Q(s_2,a'),\forall a'$, i.e., taking action $a$ is optimal in state $s_2$. Hence, (i) is proven ((ii) is proven) if (\ref{struc_prop}) holds when $a = 0$ and $a' = k^*$ ($a = k^*$ and $a' = 0$). Therefore, in the remaining, we focus on the proof of (i) while (ii) can be proven similarly. Since $s_1\geq s_2$ and based on Lemma \ref{Lemma:2}, we have $V(s_1^+) \geq V(s_2^+)$. Hence (\ref{struc_prop}) holds for $a = 0$ and $a' = k^*$, which completes the proof of (i).
\end{proof}
\begin{remark}\label{rem:1}
According to Lemma \ref{Lemma:3}, the optimal policy $\pi^*$ has a threshold-based structure, where it is optimal to transmit a status update only when the AoI/state is above some threshold value $A_{\rm th}$ (i.e., the updates are dropped when the AoI is less than or equal to $A_{\rm th}$). In addition, the scheduler uses the channel with the highest successful transmission probability (i.e., $C_{k^*}$) for each transmission attempt. Further, from (\ref{Q_func_0}) and (\ref{Q_func_a}), if $P < P_{\rm min} = \frac{\mu_{k^*}}{1 - \alpha}\left(V(2) - V(1)\right)$, then $\pi^*(s) = k^*, \forall s$, whereas if $P > P_{\rm max} = \frac{\mu_{k^*}}{1 - \alpha}\left(V(A_{\rm m}) - V(1)\right)$, then $\pi^*(s) = 0, \forall s$.
\end{remark}

Based on Lemma \ref{Lemma:3} and Remark \ref{rem:1}, $1 \leq A_{\rm th} < A_{\rm m}$ when $P \in [P_{\rm min},P_{\rm max}]$. In that case, the optimal value of the long-term average cost $\rho^*$ is obtained in closed-form in the following Lemma.
\begin{figure}[t!]
\centering
\includegraphics[width= 0.6\columnwidth]{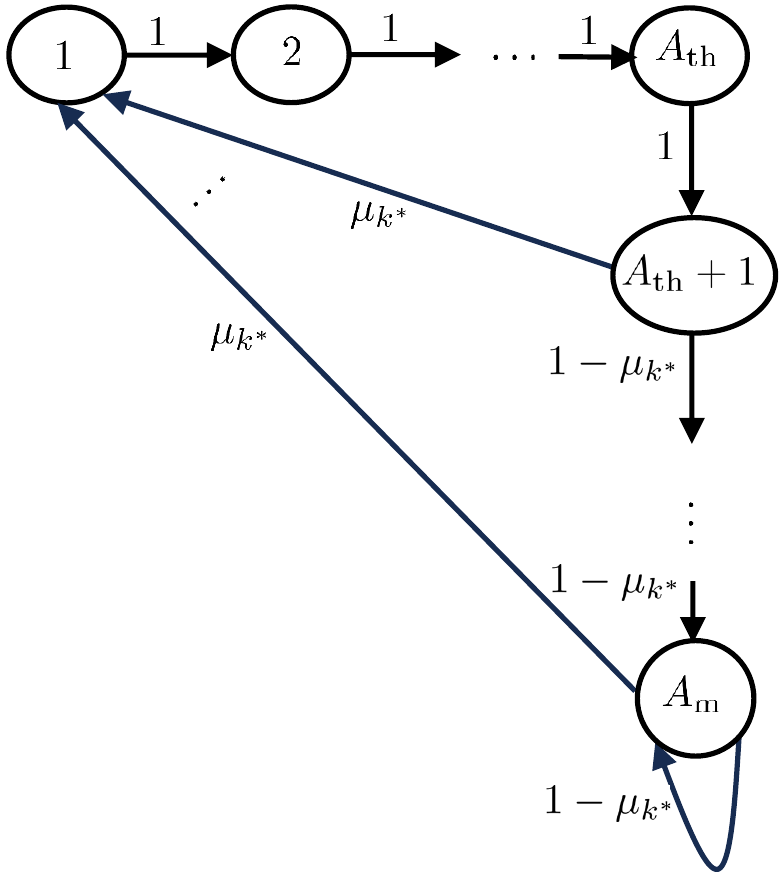}
\caption{The discrete time Markov chain induced by the optimal policy $\pi^*$.}
\label{f:MC}
\end{figure}
\begin{lemma}\label{Lemma:4}
 The optimal value of the long-term average cost associated with the optimal policy $\pi^*$ (with $1 \leq A_{\rm th} < A_{\rm m}$) is given by
 \begin{align}\label{rho_star}
\rho^* = \frac{1}{A_{\rm th} + \mu^{-1}_{k^*}} \sum_{i=1}^{3}{\beta_i},
 \end{align}
 where 
 \begin{align}
  &\beta_1 = \frac{\alpha A_{\rm th}\left(A_{\rm th} + 1\right)}{2},\\
  &\beta_2 = \frac{\alpha \left(A_{\rm th} + 1\right) + \left(1 - \alpha\right) P + \alpha \beta \left(1 - \mu_{k^*}\right)^{\beta}}{\mu_{k^*}},\\
  & \beta_3 = \frac{\alpha\left(1 - \mu_{k^*}\right)}{\mu_{k^*}^2} \left[\frac{\left(1 - \mu_{k^*}\right)\left(1 - \beta\right) - \beta}{\left(1 - \mu_{k^*}\right)^{\beta - 1} } + 1\right],\\
  & \beta = A_{\rm m} - \left(A_{\rm th} + 1\right).
 \end{align}
\end{lemma}
\begin{proof}
According to Lemma \ref{Lemma:3} and Remark \ref{rem:1}, the discrete time Markov chain (representing the system state) induced by the optimal policy $\pi^*$ is depicted in Fig. \ref{f:MC}. Thus, $\rho^*$ can be expressed as 
\begin{align}\label{rho_star_inter}
\rho^* = \sum\limits_{s \in \mathcal{S}}{C^*_s \gamma_s},
\end{align}
where $C^*_s$ is the cost of being in state $s$ under the optimal policy $\pi^*$, and $\{\gamma_s\}_{s \in \mathcal{S}}$ is the stationary distribution of the discrete time Markov chain in Fig. \ref{f:MC} (induced by $\pi^*$). Note that $C^*_s$ can be expressed as
\begin{align}\label{cost_star}
   C^*_s = \begin{cases}
   \alpha s,\;\; &\; s \leq A_{\rm th},\\
   \alpha s + \left(1 - \alpha\right) P,\;\; &\; {\rm otherwise}.
   \end{cases}
\end{align}

Thus, what remains is to obtain the probabilities $\{\gamma_s\}$. The balance equations associated with the Markov chain in Fig. \ref{f:MC} are given by
\begin{align}\label{Beqs}
    \nonumber&\gamma_s = \gamma_{s - 1},&& 2 \leq s \leq A_{\rm th} + 1,\\
    \nonumber&\gamma_s = \left(1 - \mu_{k^*}\right)\gamma_{s - 1},&& A_{\rm th} + 2 \leq k < A_{\rm m},\\
    & \mu_{k^*} \gamma_{A_{\rm m}} = \left(1 - \mu_{k^*}\right) \gamma_{A_{\rm m} - 1}.
\end{align}

From (\ref{Beqs}), we get
\begin{align}\label{Beqs_mod}
\nonumber&\gamma_i = \gamma_{A_{\rm th} + 1},\;\; 1 \leq i \leq A_{\rm th},\\
\nonumber&\gamma_{A_{\rm th} + 1 + i} = \left(1 - \mu_{k^*}\right)^i \gamma_{A_{\rm th} + 1},\;\; 0 \leq i \leq A_{\rm m} - A_{\rm th} - 2,\\
& \gamma_{A_{\rm m}} = \mu^{-1}_{k^*} \left(1 - \mu_{k^*}\right)^{A_{\rm m} - \left(A_{\rm th} + 1\right)} \gamma_{A_{\rm th} + 1}.
\end{align}

By applying $\sum\limits_{s\in\mathcal{S}}{\gamma_s} = 1$ to the set of equations in (\ref{Beqs_mod}), we obtain
\begin{align}\label{prob_th}
 \gamma_{A_{\rm th} + 1} = \frac{1}{A_{\rm th} + \mu^{-1}_{k^*}}.   
\end{align}

The final expression of $\rho^*$ in (\ref{rho_star}) is derived from substituting $\{C^*_s\}$ and $\{\gamma_s\}$ from (\ref{cost_star})-(\ref{prob_th}) into (\ref{rho_star_inter}), followed by some algebraic simplifications. This completes the proof.
\end{proof}

\section{Order-Optimal Learning Algorithms}\label{sec:learningopt}
 In this section, we use the insights obtained about the structure of the optimal policy in the previous section to develop order-optimal learning algorithms with provable finite-time guarantees for the case when the channel statistics are unknown to the scheduler. For ease of presentation of our proposed algorithms, we will use an equivalent normalized reward function $r(s(t),a(t)) \in [0,1]$ to the cost function $C(s(t),a(t))$ defined in (\ref{cost_reexpressed}). In particular, we have
 \begin{align}\label{reward_norm}
     r(s(t),a(t))=\frac{\alpha A_{\rm m} + (1-\alpha) P - C(s(t),a(t))}{\alpha (A_{\rm m}-1) + (1-\alpha) P}.
 \end{align}

According to (\ref{reward_norm}), the minimum value of $C(s(t),a(t))$, i.e., $C(1,0)$, is mapped to $r(1,0) = 1$, whereas the maximum value of $C(s(t),a(t))$, i.e., $C(A_{\rm m}, a(t) \neq 0)$, is mapped to $r(A_{\rm m}, a(t) \neq 0) = 0$. We also consider an episodic finite horizon MDP setting, where the finite time horizon $T$ is divided into $K$ episodes with equal length $H$ (i.e., each episode has $H$ time steps/slots and $T = K H$). Let $s^{\pi}_{k,h}$ and $s^{*}_{k,h}$ ($a^{\pi}_{k,h}$ and $a^{*}_{k,h}$) denote the state of the system (the action) at the $h$-th time step in episode $k$ under the learning algorithm $\pi$ and the optimal policy $\pi^*$, respectively. Note that since the channel statistics are assumed to be unknown, the transition probability matrix of the MDP under consideration is unknown to the scheduler (as can also be observed from (\ref{trans_prob})). In particular, as the scheduler interacts with the MDP over time, it observes the states, actions
and rewards generated by the unknown system dynamics (or transition probability matrix). This leads to the fundamental exploration-exploitation tradeoff where the scheduler needs to balance between exploring poorly-understood state-action pairs (to gain information and improve future performance) and exploiting its current knowledge about the system dynamics (to optimize short-run rewards). 

Before going into more details about our proposed AoI-aware order-optimal learning algorithms, it is worth noting that an AoI-agnostic learning algorithm that can handle the setting of episodic finite horizon MDPs with unknown system dynamics is the upper confidence bound value iteration (UCBVI) algorithm \cite{azar2017minimax}. The key idea of the UCBVI algorithm is to directly add an exploration bonus term (to strike a balance between exploration and exploitation) to the Q-values, rather than building confidence sets for the transition probabilities and rewards (as in UCRL2 \cite{jaksch2010near}). This leads to an improvement in the achievable regret bound by the UCBVI algorithm (compared to UCRL2). 
By directly applying the analysis of the UCBVI algorithm in \cite{azar2017minimax} to our problem, with probability $1-\delta$ ($0<\delta<1$), the regret can be upper bounded as follows
\begin{align}\label{UCBVI_regret}
R(T) \leq O\left( [ \alpha (A_{\rm m}-1) + (1-\alpha)P ]\sqrt{H A_{\rm m} C T} \right).
\end{align}

The regret bound of the UCBVI algorithm in (\ref{UCBVI_regret}) is near-optimal since it matches the established regret lower bound of \cite{jaksch2010near} for this MDP problem up to logarithmic factors: 
\begin{align}\label{regret_lowerbound}
R(T) \geq \Omega\left( [ \alpha (A_{\rm m}-1) + (1-\alpha)P ]\sqrt{H A_{\rm m} C T} \right).
\end{align}

In the sequel, we significantly improve the dependency of the regret on $T$ by developing novel AoI-aware order-optimal learning algorithms that utilize the structure of the optimal policy. In particular, our proposed learning algorithms achieve provably $O(1)$ regret bounds (i.e., the regret is bounded with respect to the increase in $T$).

\subsection{An Order-Optimal Learning Algorithm with Exploration Bonus}
As a consequence of the insights obtained in Section \ref{sec:known} for the infinite time average-cost problem, it is possible that the optimal policy (that knows the channel statistics beforehand) for the finite horizon model drops the generated status updates in certain time slots (e.g., when the AoI value is relatively small). This is in contrast to the optimal policy for the settings studied in \cite{fatale2021regret,juneja2021correlated,prasad2021decentralized,atay2021aging} (in which the transmission costs of sending status updates are ignored), where it was optimal to send a status update over the channel with the highest reliability every time slot. This key insight is utilized to develop our proposed order-optimal learning algorithms. In particular, when the action is to drop the generated status update in a certain time slot (i.e., the channels are idle), it would be useful to utilize that slot for exploring the status of one of the channels at a negligible power cost (by sending a pilot signal). 

\begin{algorithm}[t!]
	\caption{Order-optimal algorithm with exploration bonus.}
	\begin{algorithmic}
         \State \textbf{for} $k = 1, \cdots, K$ \textbf{do}
         \State \quad Compute, for all $(s,a,s') \in \mathcal{S} \times \mathcal{A} \times \mathcal{S}$,
\State \quad $N_k(s,a,s')=\sum_{\tau=1}^{k-1}\sum\limits_{s_{\tau},a_{\tau},s_{\tau}'} \mathbbm{1}(s_{\tau}=s,a_{\tau}=a,s_{\tau}'=s')$
\State \quad $N_k(s,a)=\sum_{s'\in\mathcal{S}} N_k(s,a,s')$
\State \quad $\hat{T}_k(a)=\frac{\sum_{s,s'=s^{+}}N_k(s,a,s')}{\sum_{s}N_k(s,a)}$, $\forall a \in \mathcal{A}\setminus \{0\}$
\State \quad \textbf{for} $s \in \mathcal{S}$
\State \quad \quad \textbf{if} $a=0$ \textbf{then}
   \State \quad \quad \quad $\hat{\P}_k^{'}(s^+|s,a)=1$
   \State \quad \quad \textbf{else} 
  \State \quad \quad \quad $\hat{\P}_k^{'}(s^+|s,a)=\hat{T}_k(a)$ and $\hat{\P}_k^{'}(1|s,a)=1-\hat{T}_k(a)$
  \State \quad \quad \textbf{end if}
  \State \quad \textbf{end for}
\State \quad Initialize $V_{k,H+1}(s) = 0$ for all $s \in \mathcal{S}$
  \State \quad \textbf{for} $h = H, \cdots, 1$ \textbf{do}
  \State \quad \quad \textbf{for} $(s,a) \in \mathcal{S} \times \mathcal{A}$ 
  \State \quad \quad \quad $Q_{k,h}(s,a) = \min\{Q_{k-1,h}(s,a),H,r(s,a)$ 
  \State \quad \quad \quad  \quad $+\mathbb{E}_{\hat{\P}_k^{'}}[V_{k,h+1}(s')|s,a]$
  \State \quad \quad \quad \quad $+7H\ln(5SAT/\delta)\sqrt{\frac{1}{N_k(s,a)}}\mathbbm{1}(\theta_0-k\geq 0)\},$
\State \quad \quad \quad $V_{k,h}(s)=\max_{a\in\mathcal{A}} Q_{k,h}(s,a)$
  \State \quad \quad \textbf{end for}
  \State \quad \textbf{end for}
  \State \quad \textbf{for} $h = 1, \cdots, H$ \textbf{do}
  \State \quad \quad Set $s^{\pi}_{k,1} \in \mathcal{S}$ arbitrarily
  \State \quad \quad Take action $a^{\pi}_{k,h} = {\rm arg} \underset{a \in \mathcal{A}}{\rm max}\; Q_{k,h}(s_{k,h}^{\pi},a)$ and observe 
  \State \quad \quad $s^{\pi}_{k,h+1}$
  \State \quad \quad \textbf{if} $a^{\pi}_{k,h}=0$ \textbf{then}
  \State \quad \quad Send a pilot signal over a uniformly randomly-chosen 
  \State \quad \quad channel $c$ and record $\left(s_{k,h}^{\pi},0,s_{k,h}^{\pi,+}\right)$ or $\left(s_{k,h}^{\pi},0,1\right)$ 
  \State \quad \quad based on the outcome of the transmission
  \State \quad \quad \textbf{end if} 
  \State \quad \textbf{end for}
  \State \textbf{end for} 
	\end{algorithmic}
        \label{Alg:withbonus}
\end{algorithm}
Our first order-optimal learning algorithm  is described in Algorithm \ref{Alg:withbonus}. Specifically, prior to each episode $k$, we obtain some estimates for the system dynamics (or the state transition probabilities denoted by $\{\hat{\P}'_k(s'|s,a)\}_{s,a,s'}$ based on the counts of state-action pairs experienced prior to episode $k$. These estimated transition probabilities are then fed into a backward induction algorithm with an exploration bonus term (directly added to the Q-values, similar to the UCBVI algorithm) to evaluate the policy to be executed within episode $k$. Here, the exploration bonus term is given by $7H\ln(5SAT/\delta)\sqrt{\frac{1}{N_k(s,a)}}$, which was obtained in \cite{azar2017minimax} based on the Chernoff-Hoeffding's concentration inequality. Finally, the evaluated policy is executed within episode $k$ while utilizing that whenever the action is to drop the generated status update, an opportunity for exploration is created by sending a pilot signal over a uniformly randomly-chosen channel. The achievable regret by Algorithm \ref{Alg:withbonus} is stated in the following Theorem.
\begin{theorem}\label{theorem1}
With probability $1-\delta$, the regret of Algorithm \ref{Alg:withbonus} is upper bounded as follows:
\begin{align}\label{regret_algowith}
R(T) \leq O\left( [\alpha (A_m-1) + (1-\alpha) P] H \sqrt{\theta_0} \right),
\end{align}
where $\theta_0 = \Theta(C^2 \ln \frac{2C}{\delta})$.
\end{theorem}
\begin{remark}\label{rem:2}Theorem \ref{theorem1} shows that with high probability $1-\delta$, the regret of Algorithm \ref{Alg:withbonus} with bonus terms is $O(1)$, especially that it does not increase with the total time horizon $T$ (or equivalently, the number of episodes $K$). Particularly, compared with existing MDP solutions (such as the UCBVI algorithm), we reduce the regret from $O(\sqrt{T})$ to $O(1)$, and the intuition behind that is as follows. During the early phase ($k \leq \lceil\theta_0\rceil$) where we need to carefully handle the exploration-exploitation tradeoff, the bonus term is added to encourage learning. However, thanks to the idea of sending a pilot signal when the action is to drop the generated status update, the success probabilities over different channels $\{\mu_i\}$ can be estimated much faster than the case where we do not use pilot signals (e.g., refer to (\ref{eqhoeffding}) where the use of pilot signals help us choose optimal actions with high probability in a faster way). Because of that, in the later phase ($k > \lceil\theta_0\rceil$), we can choose the optimal action greedily and with high probability. 
\end{remark}
Due to space limitations,  we will provide next a proof sketch of Theorem \ref{theorem1}.

\textit{Proof Sketch of Theorem \ref{theorem1}:} We analyze the phases before and after episode $\lceil \theta_0 \rceil$ separately. First, before episode $\lceil \theta_0 \rceil$, Algorithm \ref{Alg:withbonus} is similar to the UCBVI algorithm \cite{azar2017minimax}, but with the additional possibility of sending a pilot signal in each time step where the action is to drop the generated status update. Thus, we can still apply mathematical induction to the $V$-value function, (similar to the proof of \cite[Lemma 18]{azar2017minimax}). We first define the event 
\begin{align}
\Omega = \{ V_{k,h} \geq V_{h}^*, \forall k,h\}.
\end{align}

Under $\Omega$, all computed $V_{k,h}$ values in Algorithm \ref{Alg:withbonus} are upper bounds on the optimal
value function $V_h^* = \max_{\pi} \E[\sum_{j=h}^{H} r(s_{k,j}^{\pi},a_{k,j}^{\pi})]$. Using backward induction on steps $h$ and concentration inequalities, we can prove that $\Omega$ holds with high probability. Thus, the regret in the early phase before episode $\lceil \theta_0 \rceil$ is $\sum_{k=1}^{\lceil \theta_0 \rceil} \left[ V_{1}^*(s_{k,1}) - V'_{k,1}(s_{k,1}) \right]$, where $V'_{k,1}(s_{k,1}) = \E[\sum_{h=1}^{H} r(s_{k,h}^{\pi},a_{k,h}^{\pi})]$ represents the expected cumulative reward of Algorithm \ref{Alg:withbonus} in episode $k$. Under the event $\Omega$, we have
\begin{align}
& \sum_{k=1}^{\lceil \theta_0 \rceil} \left[ V_{1}^*(s_{k,1}) - V'_{k,1}(s_{k,1}) \right], \nonumber \\
& \leq \sum_{k=1}^{\lceil \theta_0 \rceil} \left[ V_{k,1}(s_{k,1}) - V'_{k,1}(s_{k,1}) \right], \nonumber \\
& \leq O\left( \bar{R} \sum_{k=1}^{\lceil \theta_0 \rceil} \sum_{(s,a)} H\ln(SAT/\delta)\sqrt{\frac{1}{N_k(s,a)}} \right), \nonumber \\
& \leq O\left( [\alpha (A_m-1) + (1-\alpha) P] H \sqrt{\theta_0} \right),
\end{align}
where $\bar{R} = [\alpha (A_m-1) + (1-\alpha) P]$, the second inequality is because the difference between $V_{k,1}(s_{k,1})$ and $V'_{k,1}(s_{k,1})$ is upper bounded by the bonus term $7H\ln(5SAT/\delta)\sqrt{\frac{1}{N_k(s,a)}}$ used by Algorithm \ref{Alg:withbonus}, and the last inequality is because of the pigeon hole principle. Second, after episode $\lceil \theta_0 \rceil$, we can show that with high probability, for each episode $k$, when $a_{k,h}^*=0$, $a_{k,h}^{\pi}=0$, or when $a_{k,h}^* > 0$, $a_{k,h}^{\pi} > 0$. In other words, Algorithm \ref{Alg:withbonus} will not mistakenly send or drop a status update. Specifically, according to Hoeffding's inequality, we have that after episode $\lceil \theta_0 \rceil$,
\begin{align}\label{eqhoeffding}
& \sum_{k>\lceil \theta_0 \rceil,h}\P\left\{ P(a_{k,h}^{\pi}) \neq P(a_{k,h}^{*}) \right\}, \nonumber \\
& \leq O\left( \exp{\left(-\frac{1}{2C^2}\lceil \theta_0 \rceil\right)} \right) \leq O(\delta),
\end{align}
where the last inequality is true because $\theta_0 = \Theta(C^2 \ln \frac{2C}{\delta})$. Also, according to Hoeffding's inequality and \cite[Lemma 12]{atay2021aging}, the regret due to choosing suboptimal channels after episode $\lceil \theta_0 \rceil$ can be upper bounded by $2C \cdot \left[ \exp{\left(-\frac{1}{2C^2}k\right)} + \exp{\left( - \frac{\Delta_{min}^2}{4C}k \right)} \right]$. For completeness, \cite[Lemma 12]{atay2021aging} states that for every suboptimal channel index $i\neq i^*$, the probability of choosing this suboptimal channel is bounded by
\begin{align}\label{lemma 12,13}
\P\{a_{k,h}^{\pi} = i\} \leq 2C \cdot \left[ \exp{\left(-\frac{1}{2C^2}k\right)} + \exp{\left( - \frac{\Delta_{min}^2}{4C}k \right)} \right].
\end{align}

According to (\ref{lemma 12,13}), the probability that our algorithm does not choose the optimal channel decreases in $k$ exponentially. This benefits from the fact that each episode has at least one pilot signal transmission (which occurs at the last time step of the episode where the action is to drop the generated status update, based on the implementation of the backward induction algorithm). The final regret in 
(\ref{regret_algowith}) is obtained by taking the sum over the episode index $k$.
\hfill 
\qed

\begin{figure*}[t!]
\centerline{
\subfloat[]{\includegraphics[width=0.33 \textwidth]{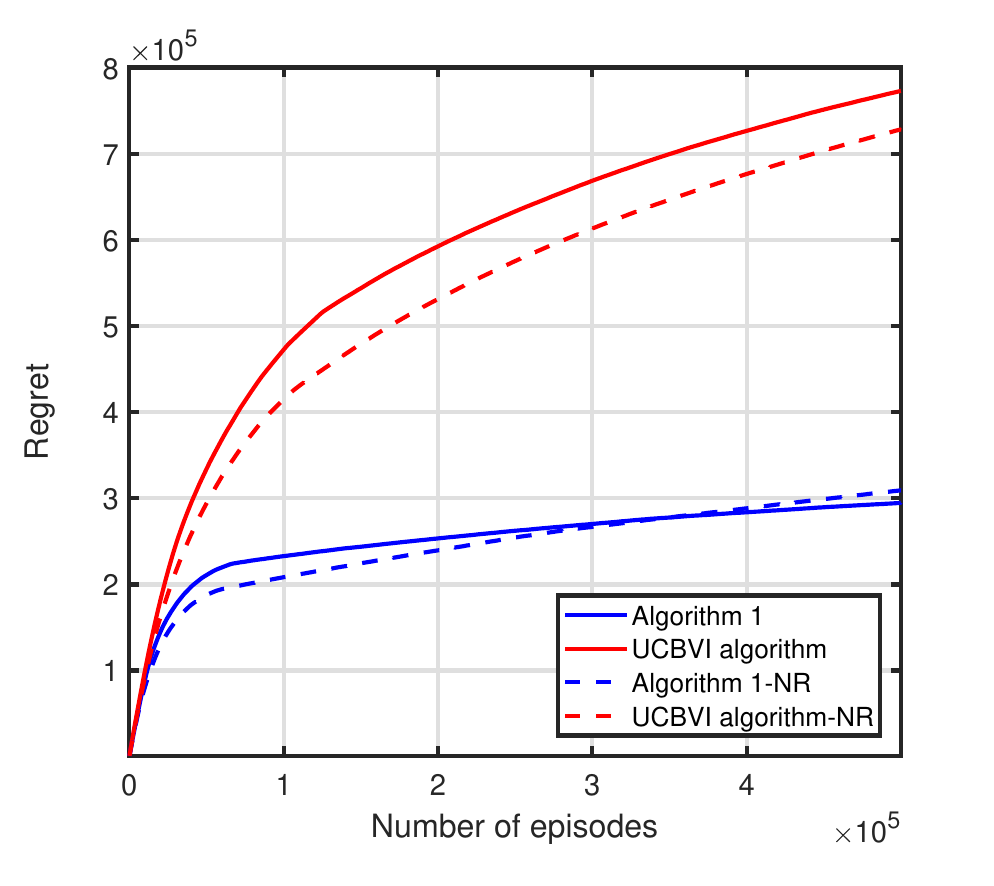}%
\label{f:h30}} \hfil
\subfloat[]{\includegraphics[width=0.33\textwidth]{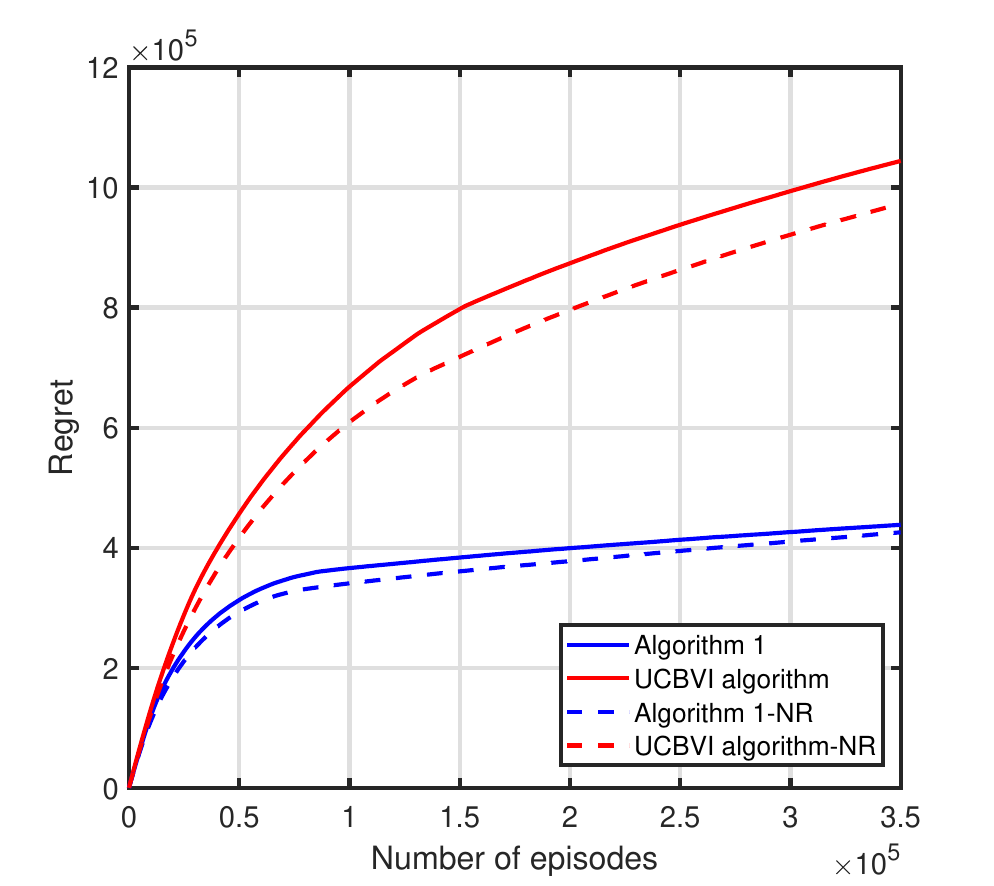}%
\label{f:h40}} \hfil
\subfloat[]{\includegraphics[width=0.33\textwidth]{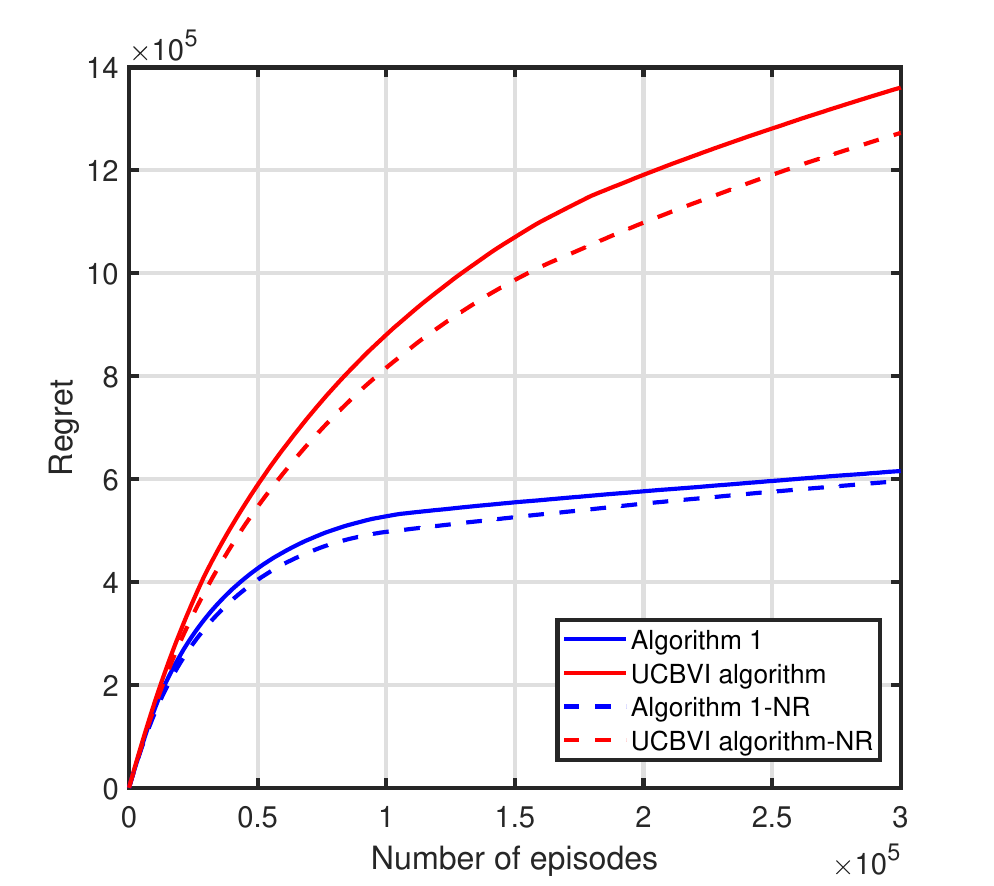}%
\label{f:h50}}}  \caption{Comparison between Algorithm \ref{Alg:withbonus} and the UCBVI algorithm. We use $A_{\rm m} = 10, P = 15, \alpha = 0.4$ and $C = 4$. The successful transmission probabilities over different channel are equally spaced from 0.2 to 0.8 (i.e., the probabilities are $\{0.2, 0.4, 0.6, 0.8\}$). We consider: i) $H=30$ in (a), ii) $H=40$ in (b), and iii) $H=50$ in (c).}\label{f:Algorithm1}
\end{figure*} 
The empirical performance of Algorithm \ref{Alg:withbonus} is compared to that of the UCBVI algorithm in Figs. \ref{f:Algorithm1} and \ref{f:importweight}. Note that the tight regret bounds in \cite[Theorems~1 and 2]{azar2017minimax} were obtained under the condition that $H \leq SA$, where $S$ and $A$ are the sizes of the state and action spaces, respectively. Thus, for fair comparison with the UCBVI algorithm, we consider values of $H$ that satisfy this condition. Also, we consider a similar exploration bonus term to that of the UCBVI algorithm (i.e., we focus on the performance of the early phase of Algorithm \ref{Alg:withbonus} in Figs. \ref{f:Algorithm1} and \ref{f:importweight}). While the initial state $s_{k,1}$ may change arbitrarily from one episode to the next \cite{azar2017minimax}, the curves with the abbreviation ``NR'' refer to the specific case where the initial state (or the initial AoI value) of each episode is set to be the AoI value at the end of its preceding episode. On the other hand, the initial state $s_{k,1}$ is chosen uniformly at random in the curves without the abbreviation ``NR''. 

A couple of key observations can be noticed from Fig. \ref{f:Algorithm1}. First, Algorithm \ref{Alg:withbonus} significantly outperforms the UCBVI algorithm in terms of the achievable regret. This demonstrates/quantifies the significant impact of the exploration opportunities created through sending pilot signals (when the channels are idle, i.e., the action is to drop the generated status update) on the achievable regret performance. Second, the achievable regret by Algorithm \ref{Alg:withbonus} approaches a bounded regret value as the number of episodes $K$ increases. However, this convergence of the regret to a bounded value appears to be relatively slow, which is mainly due to the existence of the exploration bonus term. This motivates us to develop a variant of Algorithm \ref{Alg:withbonus} with a similar provable theoretical guarantee (in terms of achieving a bounded regret with respect to $K$), but with a much better empirical performance (in terms of the fast convergence of the regret to a bounded value even when the sizes of state and action spaces are quite large). It can also be observed from Fig. \ref{f:importweight} that as the importance weight of AoI $\alpha$ increases, the gap between the achievable regrets by Algorithm \ref{Alg:withbonus} and the UCBVI algorithm slightly decreases (since the exploration opportunities of sending pilot signals become less).  

\begin{figure}[t!]
\centering
\includegraphics[width=0.7\columnwidth]{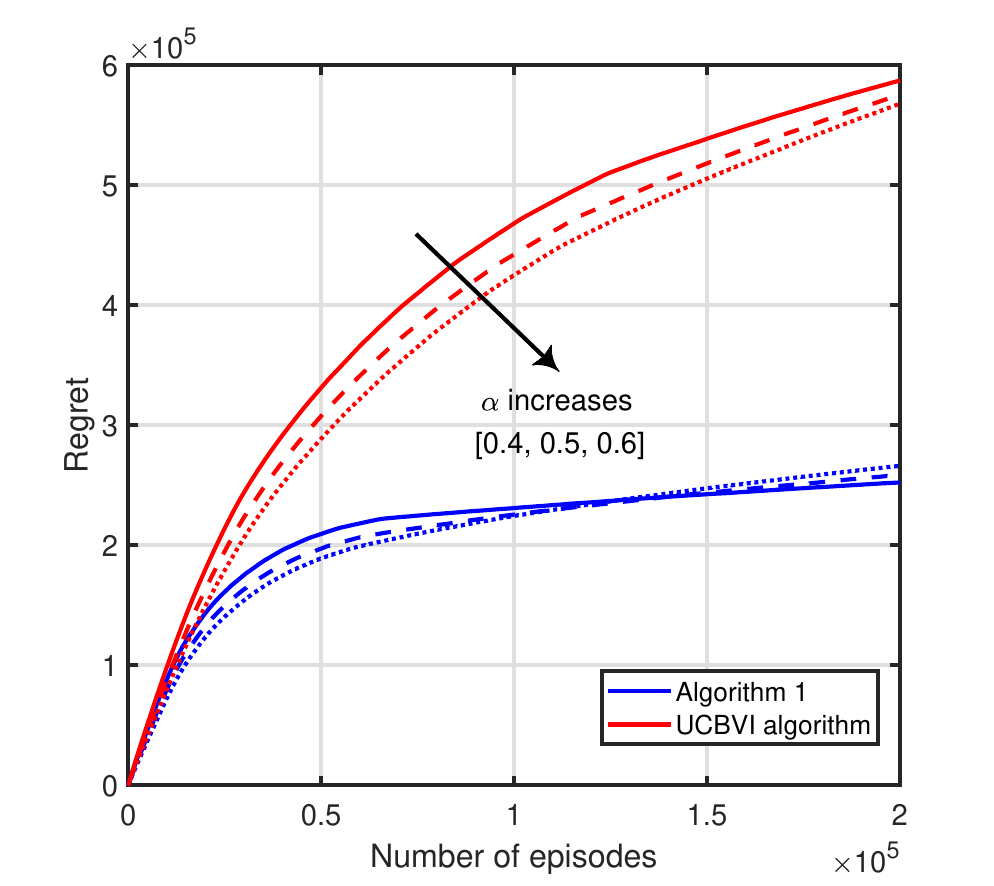}
\caption{Impact of the importance weight of AoI $\alpha$ on the achievable regret by Algorithm \ref{Alg:withbonus}. We use $H = 30$. Other parameters are same as Fig. \ref{f:Algorithm1}.}
\label{f:importweight}
\end{figure}

\subsection{An Order-Optimal Learning Algorithm without Exploration Bonus}
To overcome the limitation in the empirical performance of Algorithm \ref{Alg:withbonus} (related to the slow convergence of the regret to a bounded value), we develop another AoI-aware order-optimal learning algorithm (referred to as Algorithm 2) with a significantly better empirical performance. Different from Algorithm \ref{Alg:withbonus}, our second order-optimal learning algorithm eliminates the exploration bonus term from the Q-values. To evaluate the policy to be executed within episode $k$, this algorithm just feeds the transition probabilities estimated prior to episode $k$ into the backward induction algorithm. Specifically, the complete description of our second order-optimal learning algorithm is similar to Algorithm \ref{Alg:withbonus} with the only difference that the Q-values are evaluated as: $Q_{k,h}(s,a) = r(s,a)+\E_{\hat{\P}_k^{'}}[V_{k,h+1}(s')|s,a]$. The achievable regret by Algorithm 2 is stated in the following Theorem.
\begin{theorem}\label{theorem2}
The regret of Algorithm 2 is upper bounded as follows:
\begin{align}\label{Algo2_regret}
& R(T) \nonumber \\
& \leq O\left( H^2 A_{\rm m} P C \left[ \frac{3}{\frac{1}{2C^2}-1} + \frac{2}{\frac{\delta_{min}^2}{4C}-1} + \frac{1}{\frac{\Delta_{min}^2}{4C}-1} \right] \right),
\end{align}
where $\delta_{min} = \min_{i} \left| \frac{1-\alpha}{\alpha A_{\rm m}}P - \mu_i \right|$ and $\Delta_{min} = \min_{i} \left| \mu_{k^*} - \mu_i \right|$.
\end{theorem}
\begin{remark}\label{rem:3}Theorem \ref{theorem2} shows that the regret of our order-optimal learning algorithm without exploration bonus is $O(1)$, especially that it does not increase with the total time horizon $T$ (or equivalently, the number of episodes $K$). It is worth noting that \cite{atay2021aging} developed a learning algorithm without exploration bonus that achieves a bounded regret with respect to $T$ (i.e., $O(1)$) when the power cost of sending status updates are ignored and the status updates arrive at the source nodes according to a Bernoulli random process. Unlike the regret analysis in \cite{atay2021aging}, accounting for the power costs of sending status updates in this paper introduces several technical challenges to the regret analysis. First, the regret due to the time steps in which the learning algorithm mistakenly sends the generated status update or drops it needs to be carefully quantified. Second, in \cite{atay2021aging}, whenever the optimal channel is chosen by the learning algorithm, the regret is $0$. However, this does not hold in our setting where we consider the power cost as well. Third, different from \cite{atay2021aging} where the only reason that the optimal action is not chosen is channel estimation bias, an additional reason in our setting could be saving the power cost $P$. 
\end{remark}

Due to space limitations, we will provide next a proof sketch of Theorem \ref{theorem2}.

\textit{Proof Sketch of Theorem \ref{theorem2}}: Recall that the total cost at time step $t$ contains two parts, the AoI $s(t)$ and the power cost $P(t)$. We analyze each of them separately. First, the power regret is
\begin{align}
(1-\alpha) \sum_{k=1}^{K} \E\left[ \sum_{h=1}^{H} \left[ P(a_{k,h}^{\pi}) - P(a_{k,h}^*) \right] \right]. \label{eq2-1}
\end{align}
According to the law of total expectation, we have
\begin{align}
& \E\left[ \sum_{h=1}^{H} \left[ P(a_{k,h}^{\pi}) - P(a_{k,h}^*) \right] \right] \nonumber \\
& = \sum_{h=1}^{H} \E\left[ P(a_{k,h}^{\pi}) - P(a_{k,h}^*)  | a_{k,h}^* = a_{k,h}^{\pi} \right] \cdot \P(a_{k,h}^* = a_{k,h}^{\pi})\nonumber \\
& + \sum_{h=1}^{H} \E\left[ P(a_{k,h}^{\pi}) - P(a_{k,h}^*) | a_{k,h}^* \neq a_{k,h}^{\pi} \right] \cdot \P(a_{k,h}^* \neq a_{k,h}^{\pi}).
\end{align}

Note that when $a_{k,h}^* = a_{k,h}^{\pi}$, the power regret is $0$, i.e., 
\begin{align*}
   \E\left[ P(a_{k,h}^{\pi}) - P(a_{k,h}^*)  | a_{k,h}^* = a_{k,h}^{\pi} \right] = 0 
\end{align*}. 

Therefore, we focus on the case when $a_{k,h}^* \neq a_{k,h}^{\pi}$. Since sending a status update over any of the channels has the same power cost $P$, it is sufficient to focus on the time steps where either our learning algorithm or the optimal policy drops the generated status update. Thus, we have,
\begin{align}
& \E\left[ \sum_{h=1}^{H} \left[ P(a_{k,h}^{\pi}) - P(a_{k,h}^*) \right] \right] \leq H P \nonumber \\
& \cdot \P\left( \exists h, s.t., \{ a_{k,h}^* = 0, a_{k,h}^{\pi} > 0 \} \cup \{ a_{k,h}^* > 0, a_{k,h}^{\pi} = 0 \} \right).
\end{align}

When the estimation error $\max_{i}\{ \hat{\mu}_i - \mu_i \}$ of the transition probability, i.e., the channel reliability, is larger than the gap between the weighted AoI and power costs, the algorithm will mistakenly send the status update (instead of dropping the update and sending a pilot signal). Thus, according to Hoeffding's inequality, we have
\begin{align}
& \E\left[ \sum_{h=1}^{H} \left[ P(a_{k,h}^{\pi}) - P(a_{k,h}^*) \right] \right] ,\nonumber \\
& \leq 2HP C \cdot \left[ \exp{\left(-\frac{1}{2C^2}k\right)} + \exp{\left( - \frac{\delta_{min}^2}{4C}k \right)} \right],
\end{align}
where $\delta_{min} = \min_{i} \left| \frac{1-\alpha}{\alpha A_{\rm m}}P - \mu_i \right|$. By summing over the episode index $k$, we get the term depending on $\delta_{min}$ in the final regret.

Second, the AoI regret is
\begin{align}
\alpha \sum_{k=1}^{K} \E\left[ \sum_{h=1}^{H} s_{k,h}^{\pi} - \sum_{h=1}^{H} s_{k,h}^* \right].
\end{align}

Similarly, according to the law of total expectation, we have
\begin{align}
& \E\left[\sum_{h=1}^{H} s_{k,h}^{\pi} - \sum_{h=1}^{H} s_{k,h}^* \right] \nonumber \\
& = \sum_{h=1}^{H} \E\left[ s_{k,h}^{\pi} - s_{k,h}^* | a_{k,h}^* = a_{k,h}^{\pi} \right] \cdot \P(a_{k,h}^* = a_{k,h}^{\pi})\nonumber \\
& + \sum_{h=1}^{H} \E\left[ s_{k,h}^{\pi} - s_{k,h}^* | a_{k,h}^* \neq a_{k,h}^{\pi} \right] \cdot \P(a_{k,h}^* \neq a_{k,h}^{\pi}).
\end{align}

There are two scenarios in which we may have $a_{k,h}^* \neq a_{k,h}^{\pi}$. First, the learning algorithm (or the policy $\pi$) may mistakenly drop the generated status update to save the power cost $P$. This occurs with a probability that can be upper bounded in a way similar to what we discussed above, i.e., $2HPC \cdot \left[ \exp{\left(-\frac{1}{2C^2}k\right)} + \exp{\left( - \frac{\delta_{min}^2}{4C}k \right)} \right]$. Second, if the channel estimation error is too large, the policy $\pi$  may mistakenly choose a suboptimal channel to send a status update over. According to Hoeffding's inequality and \cite[Lemma 12]{atay2021aging}, the probability of this second scenario can be upper bounded by $2C \cdot \left[ \exp{\left(-\frac{1}{2C^2}k\right)} + \exp{\left( - \frac{\Delta_{min}^2}{4C}k \right)} \right]$. Moreover, since $s_{k,h} \leq A_{\rm m}$, we have that 
\begin{align}
\sum_{h=1}^{H} E\left[ s_{k,h}^{\pi} - s_{k,h}^* | a_{k,h}^* \neq a_{k,h}^{\pi} \right] \leq H A_{\rm m}. \label{eq2-0}
\end{align}

Since sending a status update over any of the channels has the same power cost, and both the optimal policy and $\pi$ start from the same initial state/AoI in each episode, if the event $\{ a_{k,h}^* = 0, a_{k,h}^{\pi} > 0 \} \cup \{ a_{k,h}^* > 0, a_{k,h}^{\pi} = 0 \}$ and the event $\{a_{k,h}^* \neq a_{k,h}^{\pi}, a_{k,h}^* > 0, a_{k,h}^{\pi}> 0\}$ do not occur in an episode, we have $\E\left[ s_{k,h}^{\pi} - s_{k,h}^* | a_{k,h}^* = a_{k,h}^{\pi} \right] = 0$. The final regret in (\ref{Algo2_regret}) follows by combining (\ref{eq2-1})-(\ref{eq2-0}) and taking the sum over the episode index $k$.
\hfill 
\qed

\begin{figure}
\centering
\includegraphics[width=0.67\columnwidth]{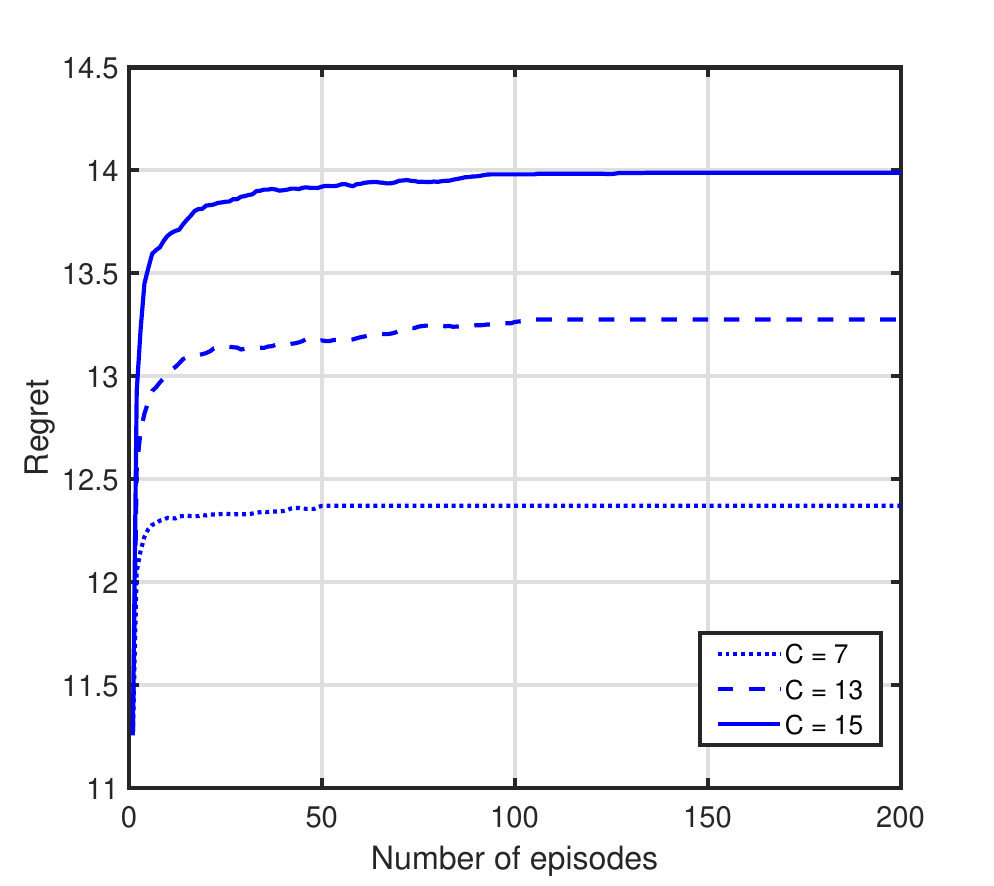}
\caption{Empirical performance of Algorithm 2. We use $A_{\rm m} = 100, P = 2, \alpha = 0.5$ and $H = 50$. The successful transmission probabilities over different channel are equally spaced from 0.2 to 0.8.}
\label{f:Algorithm2}
\end{figure}
\begin{figure}
\centering
\includegraphics[width=0.67\columnwidth]{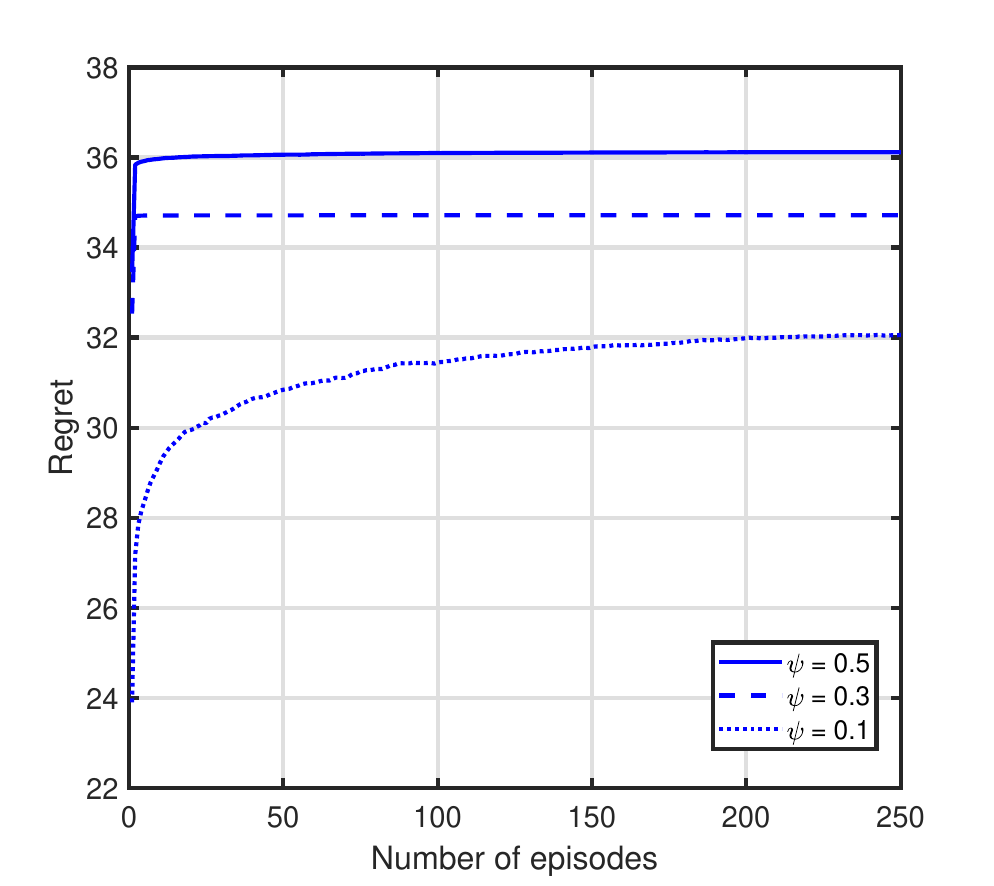}
\caption{Empirical performance of Algorithm 2 for a non-linear age function $\mathcal{F} (A(t)) = \exp{\left(\psi A(t)\right)}$. We use $A_{\rm m} = 20, C = 15, P = 2, \alpha = 0.5$ and $H = 50$. The successful transmission probabilities over different channel are equally spaced from 0.2 to 0.8.}
\label{f:Algorithm2_b}
\end{figure}
\begin{figure}
\centering
\includegraphics[width=0.7\columnwidth]{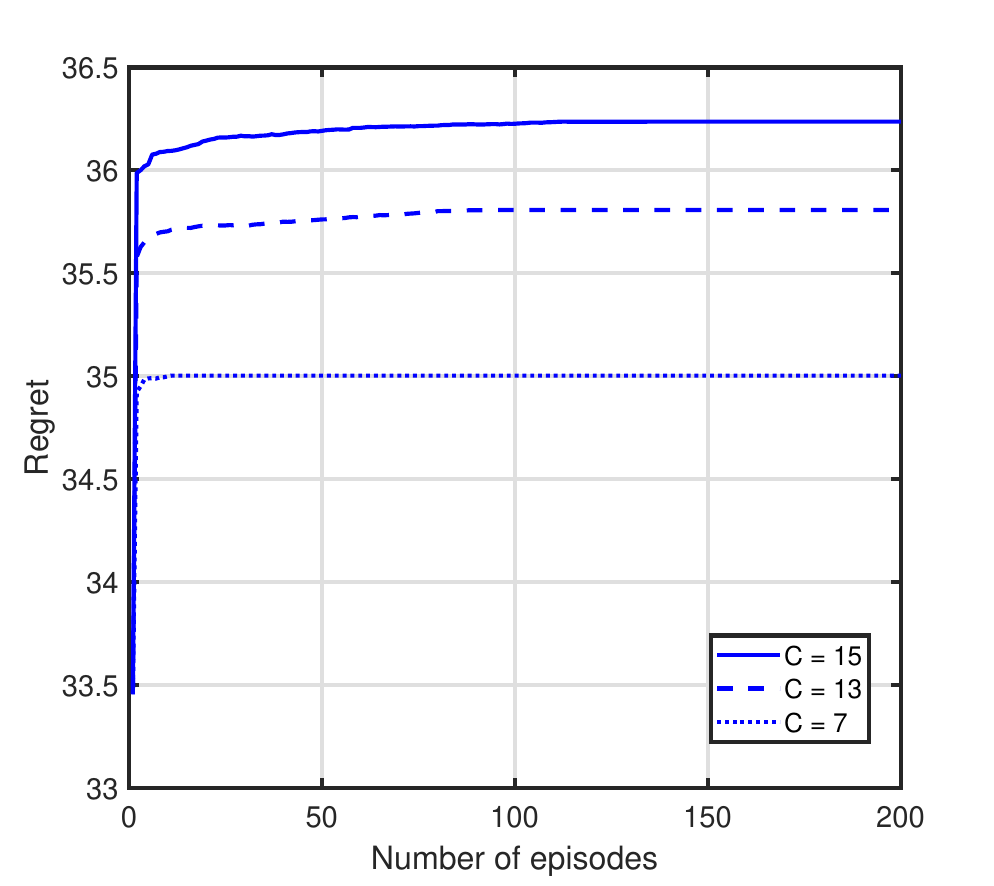}
\caption{Empirical performance of Algorithm 2 for a non-linear age function $\mathcal{F} (A(t)) = \exp{\left(\psi A(t)\right)}$. We use $A_{\rm m} = 20, \psi = 0.3, P = 2, \alpha = 0.5$ and $H = 50$. The successful transmission probabilities over different channel are equally spaced from 0.2 to 0.8.}
\label{f:Algorithm2_c}
\end{figure}
The empirical performance of Algorithm 2 is shown in Fig. \ref{f:Algorithm2}. It can be observed that the regret of Algorithm 2 (without exploration bonus) converges to a small bounded value very quickly even when the sizes of state and action spaces are relatively large. Further, it can be noticed from Fig. \ref{f:Algorithm2} that the bounded value of the regret slightly increases with the increase in the size of the action space (or equivalently, the number of channels $C$). In addition, the empirical performance of Algorithm 2 for a non-linear age function $\mathcal{F} (A(t)) = \exp{\left(\psi A(t)\right)}$ is shown in Figs. \ref{f:Algorithm2_b} and \ref{f:Algorithm2_c}. Similarly, the regret quickly converges to a small bounded value, and it can be noticed that this bounded value increases with either the rate $\psi$ of the exponential age function (Fig. \ref{f:Algorithm2_b}) or the size of action space $C$ (Fig. \ref{f:Algorithm2_c}).

\section{Conclusion}
This paper proposed novel AoI-aware online learning-based algorithms for optimizing the fundamental AoI-energy tradeoff under unknown channel statistics. In particular, we considered a system setting in which an energy-constrained source node is connected to a destination node through a set of channels, where the channel statistics were assumed to be unknown to the scheduler. For this setting, the optimal policy (that knows the channel statistics a priori) for the infinite-time average-cost problem was first proven to have a threshold-based structure with respect to the value of AoI. We then utilized this key insight to develop AoI-aware learning algorithms with a provable order-optimal regret performance for the finite time horizon model under consideration. In particular, our proposed learning algorithms with (Algorithm 1) and without (Algorithm 2) an exploration bonus were proven to surprisingly have a bounded regret performance with respect to the time horizon length (i.e., $O(1)$). 

Several system design insights were drawn from our simulation results. For instance, our results quantified the significant improvement of our proposed learning algorithms over the UCBVI algorithm in terms of the achievable regret performance. They also revealed that compared to Algorithm 1, Algorithm 2 has a much better empirical performance in terms of the fast convergence of the regret to a bounded value even when the sizes of state and action spaces are quite large. The results also showed that the bounded value of the achievable regret by Algorithm 2 slightly increases with the increase in the number of channels. 

An interesting extension of this work is to investigate the possibility of developing AoI-aware order-optimal learning algorithms for the multi-source system setting in which each source is associated with an AoI process. The study of this multi-source setting adds another layer of complexity to the analysis related to scheduling the status update transmissions from different sources. It would also be interesting to extend the proposed algorithms in this paper to account for the possibility of having: i) stochastic status update arrivals at the source node(s) \cite{atay2021aging}, and ii) time-varying unknown cost functions of AoI \cite{tripathi2021online}.

\section*{Acknowledgments}
This work has been supported in part by the Army Research Laboratory and was accomplished under Cooperative Agreement Number W911NF-23-2-0225, and by the U.S. National Science Foundation under the grants: NSF AI Institute (AI-EDGE) 2112471, CNS-2312836, CNS-2225561, and CNS-2239677. Additionally, this research was supported by the Office of Naval Research under grant N00014-24-1-2729. The views and conclusions contained in this document are those of the authors and should not be interpreted as representing the official policies, either expressed or implied, of the Army Research Laboratory or the U.S. Government. The U.S. Government is authorized to reproduce and distribute reprints for Government purposes notwithstanding any copyright notation herein.








\bibliographystyle{ACM-Reference-Format}
\bibliography{ref}











\end{document}